%% file: main.tex
\documentclass[sigconf]{acmstyles/acmart}

\usepackage{amsfonts}
\usepackage{amsthm}
\usepackage{algorithm2e}

\newtheoremstyle{customdef}
  {3pt} 
  {3pt} 
  {\itshape} 
  {} 
  {\bfseries} 
  {.} 
  { } 
  {} 

\theoremstyle{customdef}

\newcommand{\fair}{FA*IR}
\newcommand{\hyperfair}{\textit{hyperFA*IR}}

\begin{document}

\acmYear{2025}\copyrightyear{2025}
\setcopyright{rightsretained}
\acmConference[ACM FAccT '25]{ACM Conference on Fairness, Accountability, and Transparency}{June 23--26, 2025}{Athens, Greece}
\acmBooktitle{ACM Conference on Fairness, Accountability, and Transparency (ACM FAccT '25), June 23--26, 2025, Athens, Greece}
\acmDOI{10.1145/3715275.3732143}
\acmISBN{979-8-4007-1482-5/25/06}



\title{hyperFA*IR: A hypergeometric approach to fair rankings with finite candidate pool}

\author{Mauritz N. Cartier van Dissel}
\orcid{0009-0004-7916-6528}
\authornote{Corresponding authors}
\affiliation{
  \institution{Complexity Science Hub}
  \city{Vienna}
  \country{Austria}
}
\affiliation{
  \institution{Graz University of Technology}
  \city{Graz}
  \country{Austria}
}
\email{cartiervandissel@csh.ac.at}

\author{Samuel Martin-Gutierrez}
\orcid{0000-0002-5685-7834}
\affiliation{
  \institution{Complexity Science Hub}
  \city{Vienna}
  \country{Austria}
}
\affiliation{
  \institution{Graz University of Technology}
  \city{Graz}
  \country{Austria}
}
\email{martin.gutierrez@csh.ac.at}

\author{Lisette Esp\'{\i}n-Noboa}
\orcid{0000-0002-3945-2966}
\affiliation{
  \institution{Complexity Science Hub}
  \city{Vienna}
  \country{Austria}
}
\affiliation{
  \institution{Graz University of Technology}
  \city{Graz}
  \country{Austria}
}
\affiliation{
  \institution{Central European University}
  \city{Vienna}
  \country{Austria}
}
\email{espin@csh.ac.at}

\author{Ana Mar\'{\i}a Jaramillo}
\orcid{0000-0003-2409-3064}
\affiliation{
  \institution{Graz University of Technology}
  \city{Graz}
  \country{Austria}
}
\affiliation{
  \institution{Complexity Science Hub}
  \city{Vienna}
  \country{Austria}
}
\email{jaramillo@csh.ac.at}

\author{Fariba Karimi}
\orcid{0000-0002-0037-2475}
\affiliation{%
  \institution{Graz University of Technology}
  \city{Graz}
  \country{Austria}
}
\affiliation{
  \institution{Complexity Science Hub}
  \city{Vienna}
  \country{Austria}
}
\authornotemark[1]
\email{karimi@tugraz.at}

\renewcommand{\shortauthors}{Cartier van Dissel, Martin-Gutierrez, Esp\'{\i}n-Noboa, Jaramillo and Karimi}

\begin{abstract}
    Ranking algorithms play a pivotal role in decision-making processes across diverse domains, from search engines to job applications. When rankings directly impact individuals, ensuring fairness becomes essential, particularly for groups that are marginalised or misrepresented in the data. Most of the existing group fairness frameworks often rely on ensuring proportional representation of protected groups. However, these approaches face limitations in accounting for the stochastic nature of ranking processes or the finite size of candidate pools.
    To this end, we present \hyperfair, a framework for assessing and enforcing fairness in rankings drawn from a finite set of candidates.
    It relies on a generative process based on the hypergeometric distribution, which models real-world scenarios by sampling without replacement from fixed group sizes. This approach improves fairness assessment when top-$k$ selections are large relative to the pool or when protected groups are small.
    We compare our approach to the widely used binomial model, which treats each draw as independent with fixed probability, and demonstrate---both analytically and empirically---that our method more accurately reproduces the statistical properties of sampling from a finite population. To operationalise this framework, we propose a Monte Carlo-based algorithm that efficiently detects unfair rankings by avoiding computationally expensive parameter tuning. 
    Finally, we adapt our generative approach to define affirmative action policies by introducing weights into the sampling process. 
    

\end{abstract}

\maketitle

\RestyleAlgo{ruled}
\SetKwComment{Comment}{/* }{ */}

\input{Input_files/introduction}

\input{Input_files/fairness}

\input{Input_files/hyperfair}

\input{Input_files/algorithm}

\input{Input_files/affirmative_action}

\input{Input_files/conclusions}

\bibliographystyle{acmstyles/ACM-Reference-Format}
\bibliography{main}

\newpage
\input{Input_files/appendix}

\end{document}

%% file: Input_files/introduction.tex
\section{Introduction}

In our everyday lives, ranking algorithms have become ubiquitous, influencing decisions as varied as choosing the best restaurants in a city to selecting the fastest routes for a commute. However, when these algorithms rank humans---for example, candidates for jobs, university admissions, or criminal sentencing---their outcomes carry profound consequences. A person's position in such rankings can affect career opportunities \cite{geyik2019fairness}, educational access \cite{mathioudakis2020affirmative}, or even legal outcomes \cite{angwin2022machine}. Moreover, historical inequalities and data misrepresentation often lead to systematic disadvantages for certain groups, raising concerns about the fairness of ranking systems. 

In response to these challenges, recent years have seen a surge of research aimed at developing fair ranking methodologies \cite{pitoura2022fairness}. But what does it mean for a ranking to be fair? One common approach involves ensuring proportional representation of protected groups in top-ranking positions \cite{yang2017measuring, schumacher2024properties}. For instance, fairness could mean that a group appears in the first positions in proportion to its overall presence in the ranking (this approach is known as demographic parity \cite{mehrabi2021survey}). Alternatively, affirmative action policies might require fixed representation thresholds---such as balancing the representation of male and female candidates---irrespective of their population fractions. However, these fairness goals face intrinsic limitations: a ranking cannot include more members from a group than are present in the candidate pool, and the stochastic nature of the selection process can cause deviations from expected group proportions---even under fair conditions. To address this, fairness metrics must account for variability and flag rankings as unfair only when deviations are statistically significant. Determining how best to design such metrics remains an open question.

Towards this end, we introduce \hyperfair, a group fairness framework for binary attributes. Specifically, we assess fairness by testing whether a ranking could plausibly have been generated using a fair sampling strategy. A foundational approach was proposed in \cite{yang2017measuring}, which assumes that rankings are fairly generated if group membership in top positions is determined by binomial sampling with a fixed fairness probability. Building on this, \cite{zehlike2017fa} introduced a fairness-assessment strategy, where rankings are deemed unfair if they reject the hypothesis of being generated using this binomial generative model with a given confidence level. 
In our work, we propose a rigorous sampling model that corrects critical assumptions in prior methods and accounts for the variability inherent when selecting a large share of candidates or facing group imbalance.

In particular, \hyperfair~introduces the hypergeometric model, which randomly selects individuals from a finite pool with fixed counts of protected and non-protected candidates.
Unlike the binomial model, it accounts for the group of individuals already selected, ensuring that sampling probabilities evolve accordingly. To assess fairness, we focus on the group composition within the top-$k$ positions and apply hypergeometric tests that naturally reflect the finite size of the candidate pool. This approach is especially advantageous when $k$ is large relative to the population size $n$, or when the protected group is small. Indeed, it is well-established that the hypergeometric distribution diverges from the binomial once the ratio $\frac{k}{n}$ exceeds $10\%$ \cite{brunk1968teacher,johnson2005hypergeometric}. Additionally, our approach ensures fairness throughout the entire ranking, including the bottom positions---a critical but often overlooked aspect of fair rankings.

To operationalise our framework, we develop a Monte Carlo algorithm that enables fairness assessments under the hypergeometric model. Our approach supports both single tests for the top-$k$ positions and multiple tests across sequential top subsets of the ranking, ensuring a thorough evaluation of fairness. This approach resembles the one used for finding the adjusted significance level in \fair~\cite{zehlike2017fa}, but our method eliminates the need for linear searches, achieving a more efficient performance. In addition to fairness evaluation, our framework incorporates the \fair~re-ranking algorithm to adjust rankings and enforce compliance with fairness criteria.

Our key contributions are:
\begin{itemize}
    \item We highlight the mathematical and statistical boundaries when inspecting group representation across the top-k positions of a ranking. \emph{(Sections \ref{sec:fair_models} and \ref{sec:hyperfair})}
    \item We introduce a novel fair generative process based on the hypergeometric distribution that naturally applies to finite populations with fixed group sizes. \emph{(Section \ref{sec:hyper_model})}
    \item We introduce a Monte Carlo algorithm to perform multiple fairness tests across different top-ranks. \emph{(Section \ref{sec:measure})}
    \item We apply \hyperfair~on university admissions. \emph{(Section \ref{sec:university})}
    \item We extend our model to affirmative action policies, introducing a weighted sampling strategy. \emph{(Section \ref{sec:aa})}
    \item We provide a thorough comparison between the hypergeometric and binomial approaches \emph{(Section \ref{sec:hyper_model})}, and offer guidance on their respective use cases.  \emph{(Appendix \ref{app:choice})}
\end{itemize}

%% file: Input_files/fairness.tex
\section{Defining fairness in rankings}\label{sec:fair_models}

Biases and discrimination are deeply entrenched in many aspects of modern society, manifesting in areas as diverse as hiring practices \cite{raghavan2020mitigating}, criminal justice systems \cite{angwin2022machine}, and access to financial services \cite{bartlett2022consumer}. These disparities often reflect historical inequities and structural discrimination, creating challenges for fairness and equity in automated systems. In recent years, significant efforts have been devoted to developing methods to counteract these effects, particularly by designing fair algorithms in the context of automated decision-making and artificial intelligence (AI) systems \cite{ntoutsi2020bias, mehrabi2021survey}.

In the context of ranking algorithms, fairness becomes a particularly intricate issue due to the positional nature of rankings. Rankings inherently assign different levels of prominence or importance to entities based on their positions, with the visibility of items or individuals being highly skewed towards the top positions. Consequently, numerous strategies have been proposed to evaluate the fairness of rankings and develop methodologies to ensure equitable outcomes \cite{zehlike2022rank1,zehlike2022rank2, pitoura2022fairness}. Rankings can be seen as the outcome of a process, such as academic achievement or professional success. From this perspective, equality (or individual fairness) demands that individuals with similar scores or positions in the ranking process receive comparable treatment, ensuring that their visibility aligns with their underlying performance. Conversely, equity (or group fairness) emphasises fair treatment across demographic groups, even if this results in treating individuals with similar qualifications differently. This often translates to the principle of demographic parity, where the visibility of each group is proportional to its representation in the overall population \cite{espin2022inequality, pitoura2022fairness, schumacher2024properties}.

In this paper, we focus on group fairness, emphasising the principle that all groups should be treated equitably. As articulated in \cite{pitoura2022fairness}, `abstractly, a fair ranking is one where the assignment of entities to positions is not unjustifiably influenced by the values of their protected attributes'. Here, the term `unjustifiably' is very important. In some contexts, positional bias based on a protected attribute may be warranted. For example, in a ranking based on height, gender-based differences in average height could lead to justified disparities. However, for the purposes of this study, we assume that an individual’s position in a ranking should not be influenced by their membership in a protected group, and any significant deviations will be flagged as indicators of unfairness. 

Furthermore, in Sections \ref{sec:fair_models} and \ref{sec:hyperfair}, we adopt demographic parity (or statistical parity) as our fairness goal. This choice is motivated by the fact that other criteria, such as equalised odds or equal opportunity, require access to both predicted outcomes and ground-truth labels, which in the ranking setting typically correspond to underlying score values \cite{pitoura2022fairness}. Incorporating such information would require defining a utility-aware null model, which falls beyond the scope of this paper. We leave this extension for future work. 

Once fairness is defined, it can be applied to evaluate rankings or to guide re-ranking algorithms. We briefly outline both next.

\paragraph{Fairness metrics.} A large body of work has been devoted to designing metrics that assess fairness in rankings by quantifying how much a given ranking deviates from an ideal baseline. Broadly, these metrics fall into four categories: prefix metrics \cite{yang2017measuring}, which evaluate fairness in successive top-$k$ prefixes; exposure metrics \cite{singh2018fairness}, which focus on the distribution of visibility; pairwise metrics \cite{beutel2019fairness}, which compare treatment across candidate pairs; and probabilistic metrics, which assess fairness with respect to a predefined random generative model. In this paper, we concentrate on this last category. For a complete overview of fairness metrics, see \cite{raj2020comparing, pitoura2022fairness, schumacher2024properties}.


\paragraph{Re-ranking strategies.} Significant effort has also been devoted to developing re-ranking strategies, which aim to transform an unfair ranking into one that meets specified fairness criteria. These strategies are generally categorised into three types---pre-processing, in-processing, and post-processing---based on the stage of intervention in the ranking pipeline. In this paper, we focus on the post-processing approach termed \emph{FA*IR} proposed in \cite{zehlike2017fa}, which adjusts a ranking only if it fails a given fairness test. For a comprehensive overview of re-ranking algorithms and methodologies, we refer readers to the surveys by \cite{zehlike2022rank1, zehlike2022rank2, pitoura2022fairness}.

\subsection{Fairness as a generative process}\label{sec:fair_generative}

\paragraph{A running example.} To illustrate the concepts better, we introduce a running example that will be used throughout the paper. Imagine we are organisers of a conference or workshop. For this event, we can admit only a limited number of $k$ individuals. During the admission process, we receive $n$ applications—a number moderately larger than the available spots. To make admission decisions, we rank candidates based on an evaluation of their CVs and motivation letters. The top $k$ individuals in this ranking are selected. Since we have demographic information about the applicants, we decide to ensure that the accepted pool reflects the demographic distribution of the candidate population. But how can we measure this similarity? At what point do we consider the ranking unfair?

\paragraph{Top-$k$ ranking proportions.} Suppose there are a total of $n$ candidates, of which $n_p$ belong to the protected group. The ranking derived from the evaluation scores can be represented as a binary sequence $(x_1, x_2, \dots, x_n)$, where $x_i = 1$ if candidate $i$ belongs to the protected group and $x_i = 0$ otherwise. The sequence is ordered such that $x_1$ corresponds to the highest-ranked candidate and $x_n$ to the lowest. The cumulative count of protected group members in the top $k$ positions is denoted by $y_k = \sum_{i=1}^k x_i$, while the proportion of protected candidates in the top $k$ is given by $p_k = \frac{y_k}{k}$.

Fairness is usually assessed by comparing $p_k$, the proportion of protected candidates in the top $k$, to $p = \frac{n_p}{n}$, the overall proportion of the protected group in the population \cite{yang2017measuring, schumacher2024properties}. A ranking can be considered fair if $p_k$ remains close to $p$ across the ranking. However, this comparison is subject to several important constraints. At small $k$, $p_k$ exhibits high variability due to the limited number of candidates considered. For instance, $p_1$ can only take values of $0$ or $1$. Additionally, the finite size of the candidate pool imposes inherent boundaries on $p_k$. These boundaries correspond to two extreme scenarios. Let $x = \frac{k}{n}$ denote the proportion of individuals considered up to rank $k$. In the first scenario, all protected candidates occupy the top positions, resulting in $p_k = 1$ until $x = p$, after which $p_k$ follows the curve $\frac{p}{x}$. In the second scenario, all protected candidates are relegated to the lowest positions, leading to $p_k = 0$ until $x = 1-p$, after which $p_k$ follows the curve $1 - \frac{1-p}{x}$. These boundaries, which arise from the constraints of a finite population, are discussed in greater detail in Appendix \ref{app:boundaries}, and are visualised in panel (b) of Figure \ref{fig:diagram}, for the case $p=0.3$.

Finally, randomness plays a critical role in determining deviations from fairness. Even under fair conditions, random sampling can result in deviations between $p_k$ and $p$, particularly for small $k$. For instance, in the top three positions, it is plausible to see only one demographic group represented, even if both groups are equally sized. To quantify such deviations, it is useful to introduce a null model based on random sampling and compare the observed $p_k$ to its expected behaviour under this model. We now describe a commonly used sampling model.

\subsection{The binomial model}
In \cite{yang2017measuring} the authors define a fair ranking as one generated through the following probabilistic process. Given a set of ranked items, a fair ranking is created by initialising an empty list and sequentially adding items from the observed ranking. For each position $j$, a Bernoulli trial with probability $f$ determines whether to select the most highly ranked available item from the protected group (if the trial succeeds) or from the non-protected group (if it fails), for as long as there are items in both groups. The probability $f$ is typically set to the proportion of the protected group within the overall population when the goal is demographic parity. A simple example showcasing this generative process can be found in panel (a) of Figure \ref{fig:diagram}, with $n=10$ and $n_p=3$. We now present an equivalent and more detailed definition for such model. 

\begin{definition}[Finite binomial model]
A ranking $(x_1, x_2, \dots, x_n)$ with $n_p$ protected elements is said to be generated under the \emph{finite binomial model} with fairness probability $f \in [0,1]$ if, for all $i \leq n$ and as long as candidates remain for both groups, $x_i$ is a realisation of the Bernoulli random variable $X_i \sim \emph{B}(f)$, where all $X_i$ are independent and identically distributed (iid). If, after a certain position $t$, only candidates from one group remain, then $x_i$ is deterministically assigned to that group for all $i \geq t$. Let $Y_k = \sum_{i=1}^k X_i$ denote the cumulative count of protected individuals among the top $k$ positions. For small values of $k$, specifically $ k < \min(n_p, n - n_p) $, $Y_k$ follows a binomial distribution of the form $\emph{Bin}(f,k)$.

\end{definition}

In \cite{zehlike2017fa}, the authors approximate this model using a standard binomial distribution. This approximation holds when the number of selected candidates $k$ satisfies $k < \min(n_p, n - n_p)$, as the depletion of the population has no effect under these conditions. However, as $k$ grows larger, $Y_k$ diverges from the binomial assumption, necessitating a more nuanced interpretation. For this reason, we refer to this framework as the finite binomial model.

\begin{figure*}[ht]
\begin{center}
\includegraphics[width=0.9\textwidth]{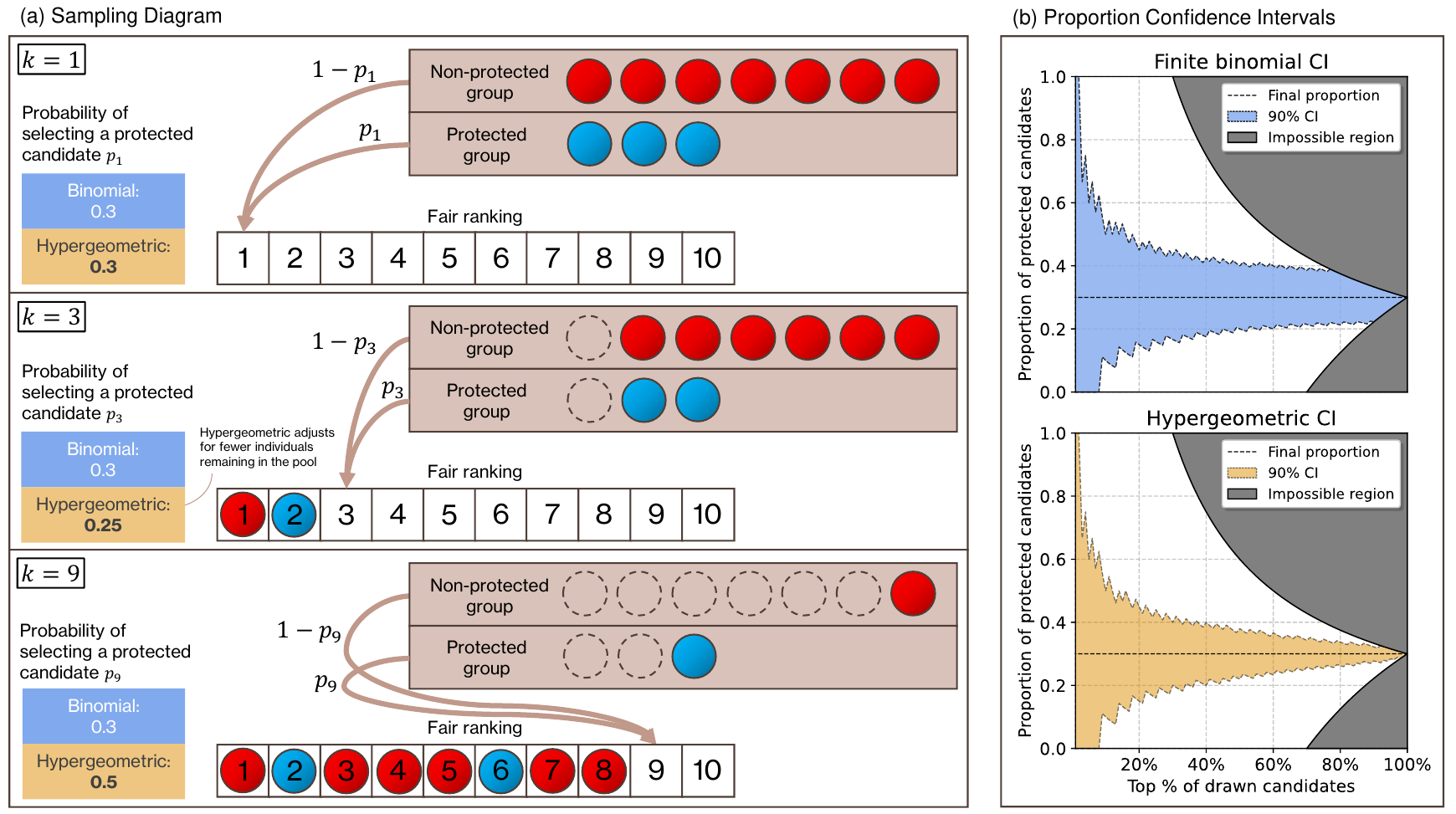}
\caption{\textbf{Binomial vs. Hypergeometric sampling in fair rankings.} Panel (a) applies these concepts to a conference admission scenario, in which the organisers aim to admit $k$ individuals from a pool of $n = 10$ applicants, consisting of 3 protected (blue) and 7 non-protected (red) candidates, while ensuring demographic fairness. For $k=1$, the probabilities of selecting a protected individual are identical under both models, as the finite pool constraints have no impact. For $k=3$, the hypergeometric probability decreases compared to the binomial model, reflecting the depletion of protected individuals in the pool. For $k=9$, the hypergeometric probability increases relative to the binomial, as most non-protected individuals have already been selected. This demonstrates the hypergeometric model's ability to dynamically adjust sampling probabilities, making it a more realistic approach for fair candidate selection in rankings. Panel (b) visualises the distribution of the proportions of protected individuals as the percentage of drawn individuals increases, for both the finite binomial and the hypergeometric models. Particularly, it shows the two-sided $90\%$ confidence intervals, with $n=100$, $n_p=30$, and fairness probability $f=\frac{n_p}{n}=0.3$. The hypergeometric model holds tighter confidence intervals, especially when $k$ approaches $n$. The grey area illustrates the region of proportions of protected candidates that cannot be reached because of the finite population.}
\Description{Figure comparing binomial and hypergeometric sampling in the context of ranking fairness. Panel (a) illustrates differences in protected group selection probabilities for different top $k$s in a conference admission example. Panel (b) shows confidence intervals for group proportions under each model, highlighting how hypergeometric sampling accounts for finite candidate pools.}
\label{fig:diagram}
\end{center}
\end{figure*}

The top subfigure in panel (b) of Figure \ref{fig:diagram} depicts the two-sided 90\% confidence intervals (CIs) for the proportion of protected individuals selected under the finite binomial model. These proportions are shown as a function of the percentage of individuals drawn relative to the total population size, with fixed parameters $n = 100$, $n_p = 30$, and $f = p = 0.3$. This visualisation of top-$k$ ranking proportions aligns with prior studies, such as \cite{espin2022inequality, stoica2024fairness}, which employ similar techniques to analyse ranking algorithms like PageRank. Furthermore, the gray regions represent the mathematical boundaries of these proportions, offering a novel perspective on the constraints imposed by finite populations, as discussed earlier.


\paragraph{The limitations of the binomial model.} When $f$ is equal to $p$, which means that our goal is to ensure demographic parity, the standard binomial model can be understood as sampling with replacement from a pool of $n_p$ protected and $n-n_p$ non-protected individuals, where, after each draw, the selected item is reintroduced into the pool of possible items. This ensures that the probability of selecting a protected candidate remains constant throughout the process. However, in our finite binomial model, this assumption of independence breaks down as the actual pool of candidates begins to deplete. Specifically, once the available number of candidates runs low, dependencies arise because it becomes impossible to sample more individuals than exist in the population. For instance, if there are only three protected individuals in the population, the probability of selecting a fourth protected individual is zero---an outcome inconsistent with a pure binomial distribution.

This effect is visually evident in panel (b) of Figure \ref{fig:diagram}. As the percentage of drawn individuals increases and approaches the impossible regions (depicted in gray), the finite binomial model transitions into a constrained regime. Importantly, reaching these boundaries has significant implications: it forces all remaining positions (the lower-ranked ones) to be filled by candidates from only one group. While this might not seem immediately unfair for top-ranked positions, it leads to disproportionate representation in the bottom ranks, which is clearly problematic and discriminatory.

A more intuitive and accurate approach is to sample without replacement, where the probability of selecting a particular individual dynamically adjusts based on the candidates already drawn. This method better reflects real-world scenarios and is less likely to encounter the artificial constraints imposed by the boundaries. We explore this alternative in the next section.

%% file: Input_files/hyperfair.tex
\section{\hyperfair: a hypergeometric framework for fair rankings}\label{sec:hyperfair}


\subsection{The hypergeometric model}\label{sec:hyper_model}


The hypergeometric null model provides a more precise approach for modelling fairness in rankings when the population is finite, and each draw is made without replacement. Unlike the binomial model, which assumes independence between draws, the hypergeometric model accounts for the dependency between successive draws as the population is depleted. As with the finite binomial model, the process begins with a ranking and an empty list. However, instead of using a fixed probability, the sampling probability dynamically adjusts based on the number of remaining candidates. Specifically, suppose that after $j$ draws, $y_j$ protected candidates have already been selected. At the $j+1$ draw, the probability of selecting a protected individual is given by the number of remaining protected candidates over the number of all remaining candidates, i.e. $\frac{n_p-y_{j}}{n-j}$, where $n_p$ is the number of protected candidates and $n$ is the size of the total candidate pool. 

Figure \ref{fig:diagram} (a) illustrates this process in the context of our conference admission example, where the goal is to select $k$ individuals from a pool of 10 candidates, 3 of whom belong to a protected group. Initially, the probabilities under the binomial and hypergeometric models are identical. However, as candidates are admitted, the probabilities under the hypergeometric model adapt dynamically to reflect the finite pool of applicants. For example, after admitting two candidates---one protected and one non-protected---the probability of admitting a protected candidate on the third draw becomes $\frac{1}{4}$, reflecting the two remaining protected candidates out of the total eight remaining applicants. Similarly, by the ninth draw, when only one candidate from each group remains, the probability of selecting the protected candidate rises to $\frac{1}{2}$. This dynamic adjustment aligns with the finite nature of the applicant pool, making the hypergeometric model more appropriate for scenarios such as this conference admission process. We now formally define this model:


\begin{definition}[Hypergeometric model]
A ranking $(x_1, x_2, \dots, x_n)$ with $n_p$ protected elements is said to be generated under the \emph{hypergeometric model} if we draw each element in the ranking uniformly at random from a population of size $n$ that contains exactly $n_p$ protected individuals. The sampling is performed without replacement, and the probability of selecting a protected individual adjusts dynamically based on the remaining population. Let $Y_k=\sum_{i=1}^k X_i$, representing the cumulative number of protected individuals in the top $k$ positions. Then $Y_k$ follows a hypergeometric distribution of the form $\emph{Hyp}(n,n_p,k)$.
\end{definition}

The hypergeometric model can also be interpreted as a random permutation of a population consisting of $n_p$ protected candidates and $n-n_p$ non-protected candidates. This equivalence arises because sampling without replacement inherently randomises the order of elements while maintaining the constraints of the original population. In contrast, the standard binomial model, which assumes sampling with replacement, does not correspond to a random permutation. Instead, each draw is independent of the others, and the total number of protected or non-protected individuals in the pool is not fixed, making it possible to deviate from the initial population proportions. The finite binomial model occupies an intermediate position: while it starts with near-independent sampling, it transitions to a constrained regime as the pool depletes. In this constrained phase, the process partially aligns with the hypergeometric model, reflecting some dependency among draws, though not fully converging to a random permutation. 


\paragraph{The hypergeometric model leads to more restrictive tests.} Figure \ref{fig:diagram} (b) bottom figure illustrates the two-sided 90\% confidence intervals (CIs) for the percentage of protected individuals drawn under this hypergeometric model, with $n=100$ candidates, of which $n_p=30$ are protected. It is clear that when the share of drawn individuals $\frac{k}{n}$ is small, the hypergeometric and the binomial distribution are very similar. In fact, it's a well known fact that if $ k $ and $ p $ are fixed, the hypergeometric distribution converges in distribution to the binomial as $ n \to \infty $ (with $ n_p = pn $). 

However, as $ k $ increases and becomes a larger fraction of $ n $ (e.g., $ \frac{k}{n} > 0.1 $), the two distributions start to diverge \cite{brunk1968teacher, johnson2005hypergeometric}. This is due to the dependency between draws in the hypergeometric case, which reduces the variance relative to the independent sampling assumed in the binomial case. The hypergeometric distribution becomes more concentrated around its mean, leading to tighter confidence intervals.  In particular, Uhlmann \cite{uhlmann1966vergleich} showed that for outcomes $ y $ below or above the mean $\bar y= pk $, the hypergeometric probability is smaller than the binomial \cite{dudley2007uhlmann}:  
\[
\mathbb P(H_k = y) \leq \mathbb P(B_k = y) \quad \text{for } y \leq pk - 1 \text{ and } y \geq pk + 1.
\]  
where $H_k$ and $B_k$ denote the random variables representing the number of protected individuals among the top-$ k $ positions respectively for the hypergeometric $\text{Hyp}(n, n_p, k) $ and for the binomial $ \text{Bin}(p, k) $ distribution.  As a result, the cumulative distribution function of the hypergeometric is lower than that of the binomial for $ y \leq pk - 1 $, leading to more restrictive fairness tests (see Section \ref{sec:measure}), and consequently demanding for higher fairness requirements when $k$ grows larger. 

\paragraph{The hypergeometric model is more accurate when the number of protected individuals is particularly small.} As emphasised in manufacturing standards for attribute acceptance sampling \cite{samohyl2018acceptance}, where it is used to decide whether to accept or reject a batch, the hypergeometric model is particularly useful when $n_p$ (or $n - n_p$) is small. For instance, consider a scenario where $n = 10$ and $n_p = 1$, meaning there is only one protected candidate in the population. In the hypergeometric model, as non-protected candidates are drawn, the probability of sampling the single protected candidate increases with each draw, reflecting the depletion of non-protected individuals. For example, after drawing 5 non-protected candidates, the probability of selecting the protected candidate becomes $\frac{1}{5}$. In contrast, under the finite binomial model, the probability of sampling the protected candidate remains fixed at $\frac{1}{n} = \frac{1}{10}$, regardless of the number of non-protected candidates already selected, which can lead to unrealistic outcomes.

\paragraph{The hypergeometric model considers fairness in the bottom positions.}
Finally, the hypergeometric distribution also ensures that fairness is respected in the bottom positions. Ensuring fairness at the bottom of a ranking is crucial in many real‑world applications where opportunities, resources, or risks are distributed across the entire ranking, not just the top positions. For instance, in hiring pipelines or college admissions, candidates ranked lower may still benefit from waiting lists or fallback options. Similarly, in public service delivery or housing allocation, fairness at the bottom helps prevent systemic disadvantage and ensures equitable treatment for under-represented groups.

Let $H_k\sim\mathrm{Hyp}(n,n_p,k)$ denote the number of protected individuals in the top‑$k$ positions, and let $W_t$ denote the number in the bottom‑$t$. By exchangeability of hypergeometric sampling we have $\mathbb P(H_k=h)=\mathbb P(W_k=h)$ for all values of $h$ and thus their fairness quantiles coincide. Moreover, since $H_k + W_{\,n-k}=n_p$ it follows that $H_k=n_p-W_{\,n-k}$, which shows that testing under‑representation in the top‑$k$ is equivalent to testing over‑representation in the bottom‑$(n-k)$, and a two‑sided test on one maps exactly to a two‑sided test on the other.  Consequently, verifying fairness at the top‑$k$ automatically guarantees fairness at the bottom‑$(n-k)$. We provide a formal proof of this statement in Appendix \ref{app:proof}.


Under the finite binomial model, by contrast, one tests whether the top‑$k$ count $Y_k$ follows $\mathrm{Bin}(k,f)$.  There is no exchangeability between head and tail, so the bottom‑$(n-k)$ count does not follow the same binomial law. Hence, passing a binomial test at the top carries no implication for the tail, and a separate bottom‑$(n-k)$ assessment is required.

%% file: Input_files/algorithm.tex
\subsection{Testing for fairness}\label{sec:measure}

Building on the characteristics of fair rankings under random sampling, we now propose a method to assess whether an observed ranking is fair. While this section focuses on addressing the under-representation of the protected group, the proposed methods are versatile and can be adapted to other fairness objectives, such as mitigating over-representation or ensuring balanced representation across both protected and non-protected groups.

\paragraph{Single statistical test.} Consider our ranking $(x_1, x_2, \dots, x_n) $, where $ y_k = \sum_{i=1}^k x_i $ represents the cumulative number of protected individuals in the top $ k $ positions. To evaluate fairness, we compare the observed ranking to the hypergeometric model introduced earlier. A straightforward approach is to perform a single hypergeometric test to determine whether the observed number of protected individuals, $ y_k $, could plausibly result from $ k $ successive draws without replacement under the null model. 

Since we are concerned with under-representation of the protected group, we conduct a one-sided test to check whether the probability of observing $ y_k $ under the hypergeometric null model falls below a predefined significance level $ \alpha $. Formally, we assume under the null hypothesis that $ y_k $ is a realisation of the random variable $ Y_k $, which follows a hypergeometric distribution. Let $ F_k(\cdot) $ denote the cumulative distribution function (CDF) of $ Y_k $\footnote{$F_k(y)=\mathbb P(Y_k\le y_k)=\sum_{j=0}^{y} \frac{\binom{n_p}{j} \binom{n-n_p}{k-j}}{\binom{n}{k}} $}. We reject the null hypothesis if the p-value is smaller or equal than $\alpha$, i.e.:
\[
 \mathbb P(Y_k\le y_k) = F_k(y_k) \le \alpha.
\]
If $ H_0 $ is rejected, we conclude that the ranking prefix from position $ 1 $ to $ k $ is unfair. This approach provides a simple yet effective mechanism to evaluate fairness at a specific position in the ranking.

\paragraph{Multiple statistical tests.} In some cases, it is important to assess whether fairness holds not only at a single cutoff $k$ but simultaneously for every top‑$j$ subset with $1\le j\le k$. This is particularly relevant when analysing fairness throughout the entire ranking, as a single test at $ k = n $ would be uninformative---by definition, the total number of protected individuals sampled will be $ n_p $. However, challenges arise because fairness must be maintained across all top-$j$ subsets, inducing dependencies between the $k$ tests.

To control the family‑wise error rate at level $\alpha$, we replace the significance $\alpha$ by a stricter threshold $\alpha_c$ for each individual test.  Concretely, for each $j=1,2,\dots,k$ we test
$$
H_{0,j}:\;Y_j\sim \text{Hyp}(n,n_p,j)
\quad\text{vs.}\quad
H_{1,j}:\;Y_j\text{ is too small,}
$$
by computing the one‑sided p‑value $F_j(y_j)$.  We then reject the test for any $H_{0,j}$ for which $F_j(y_j)\le \alpha_c$. To find $\alpha_c$, we can observe that the probability of failing at least one of the $k$ tests corresponds to the complement of the probability of success across all tests:
\[
\mathbb P(\text{fail}) = 1 - \mathbb P(F_1(Y_1) > \alpha_c, \dots, F_k(Y_k) > \alpha_c).
\]
By defining $Z_k = \min_{i=1,\dots,k} F_i(Y_i)$, we can rewrite the failure probability as:
\[
\mathbb P(\text{fail}) = 1 - \mathbb P(Z_k > \alpha_c) = \mathbb P(Z_k \leq \alpha_c).
\]
The correction $\alpha_c$ is then chosen such that this probability is bounded by $\alpha$. More precisely, we choose the adjusted parameter to be the largest value that satisfies this condition, i.e.
$$
\alpha_c := \sup\{\gamma\in[0,1]:\mathbb P(Z_k\le \gamma)\le\alpha\}
$$
In other words, $\alpha_c$ is defined as the upper $\alpha$-quantile of the distribution of $Z_k$. Without this adjustment ($\alpha_c = \alpha $), the cumulative rejection probability increases significantly, as illustrated by the blue line in Figure \ref{fig:algo} (a).

In the binomial case, \cite{zehlike2017fa} addressed this problem by developing an algorithm that computes the probability of accepting all hypotheses under a given significance threshold analytically, using the inverse binomial CDF. They then employed a linear search algorithm to find  $\alpha_c $ such that the probability of rejecting at least one test equals  $\alpha $ \cite{zehlike2017fa, zehlike2020note}. However, applying this framework to the hypergeometric distribution is infeasible due to the increased complexity introduced by dependencies in the hypergeometric case. Consequently, we propose an alternative Monte Carlo-based algorithm for fairness evaluation under the hypergeometric null model.

\paragraph{The Monte Carlo algorithm.} 

\begin{figure*}[ht]
\begin{center}
\includegraphics[width=\textwidth]{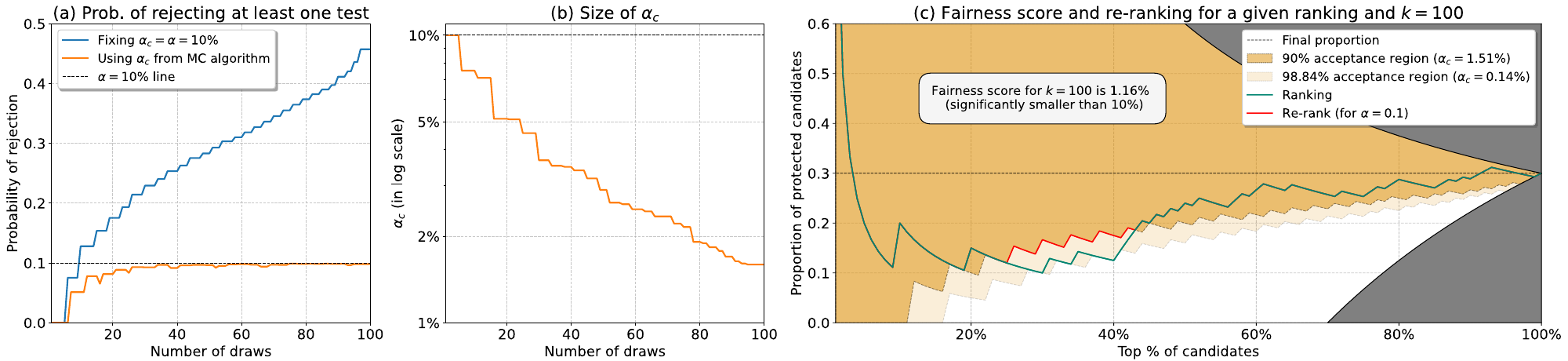}
\caption{\textbf{Evaluation of fairness under multiple hypothesis testing ($n=100,n_p=30$).} (a) Probability of failing at least one of the multiple tests versus the number of draws, shown for both unadjusted and adjusted significance levels. The adjusted $\alpha_c$ ensures that the overall failure probability is controlled at $\alpha=10\%$. (b) The size of the adjusted significance level ($\alpha_c$) as a function of the number of draws, demonstrating how $\alpha_c$ decreases with an increasing number of tests to maintain the desired overall failure probability. (c) Fairness evaluation and re-ranking for a ranking system with $ k = 100 $. The plot illustrates the proportion of protected individuals in the top of the ranking (green line) as a function of the percentage of individuals drawn. The fairness score $\tilde f$ for $ k = 100 $ is $1.16\%$, indicating substantial unfairness. The lighter orange region represents this metric by showing the first confidence interval (CI) where the observed ranking would pass the fairness test. For a higher significance level of $\alpha = 10\%$ (with an adjusted $\alpha_c = 1.59\%$), the ranking fails the test as it exceeds the boundaries of the orange CI. In this case, re-ranking the individuals adjusts the ranking to align with the boundary of the $10\%$ CI, as shown by the red curve.}
\Description{Figure 2 evaluates fairness under multiple hypothesis testing. Panel (a) shows the probability of failing at least one test as the number of top-k tests increases. The unadjusted case exceeds the significance level $\alpha = 0.1$, while the adjusted case controls the error. Panel (b) shows how the corrected significance level $\alpha_c$ decreases with the number of tests. Panel (c) displays a fairness evaluation on a ranking with $n=100$ and $n_p=30$, where the observed curve violates the acceptance region corresponding to $\alpha_c = 1.53\%$. The re-ranked curve (red) restores fairness at level $\alpha = 10\%$.}
\label{fig:algo}
\end{center}
\end{figure*}

Due to the dependencies among the random variables $Y_i$, it is challenging to compute the distribution of $Z_k$ analytically.  We instead estimate it using the following Monte Carlo procedure:
\begin{enumerate}
    \item Generate $N_e$ random ranking permutations of length $n$ containing exactly $n_p$ protected individuals and $n-n_p$ non protected ones. 
    \item For each ranking and for each position $i = 1, \dots, k$, compute the cumulative sum $y_i$, and evaluate $F_i(y_i)$ either analytically (using the hypergeometric CDF) or empirically (using quantile positions relative to other generated rankings).
    \item Compute $z_k = \min_{i=1,\dots,k} F_i(y_i)$ for each ranking.
    \item Estimate the distribution of $Z_k$ from the generated $z_k$ values and compute $q_{Z_k}(\alpha)$ to find $\alpha_c$.
\end{enumerate}

The choice of $N_e$, the number of random rankings, is crucial and should be a number sufficiently large to ensure the convergence of the algorithm. A brief discussion regarding the algorithm's computational cost and convergence can be found in Appendix \ref{app:convergence}. 

Notice that the Monte Carlo algorithm enables quick recalculations of the adjusted significance level $\alpha_c$ for varying values of $k$ by simply modifying the index over which the minimum is applied. Although the algorithm is demonstrated here using the hypergeometric null model, it can also be adapted to compute the corrected parameter for the finite binomial model by modifying the generated random rankings. When applied to the binomial case, this approach is often faster than the method proposed in \cite{zehlike2017fa}, as it avoids relying on a linear search to determine the adjusted parameter.

Figure \ref{fig:algo} (a) shows the result of the $\alpha_c$ adjustment where the orange line demonstrates that the overall probability of rejection remains controlled and close to the target significance level $\alpha$. Panel (b) shows the adjusted significance level $\alpha_c$ obtained using our Monte Carlo as a function of the number of draws. When the probability of rejection is zero (in the first draws), we set $\alpha_c=\alpha$. It is clear that the size of the adjusted parameter is monotonically decreasing as more tests are performed.

\subsection{A probabilistic fairness metric}


The Monte Carlo approach that has been used to define $Z_k$ can also be used to define a fairness metric for a given ranking. Specifically, given an observed ranking $(\tilde{x}_1, \dots, \tilde{x}_n)$, we calculate the cumulative sums $\tilde{y}_i = \sum_{j=1}^i \tilde{x}_j$ for $i = 1, \dots, k$. Using these values, we compute $\tilde{z}_k = \min_{i=1,\dots,k} F_i(\tilde{y}_i)$, where $F_i$ is the CDF calculated either analytically or empirically. The fairness metric, $\tilde{f}_k$, is then defined as the probability that $Z_k \leq \tilde{z}_k$ under the null distribution:
\[
\tilde{f}_k = \mathbb P(Z_k \leq \tilde{z}_k).
\]
In practice, this p-value is computed by evaluating the proportion of random rankings generated under the null model where $z_k$ is less or equal than $\tilde{z}_k$. This metric reflects the likelihood of observing a ranking as extreme as the given one under the fair hypergeometric null model, and therefore a smaller $\tilde{f}_k$ indicates greater deviation from fairness. Figure \ref{fig:algo} (c) illustrates the score of our fairness metric applied to a scenario where the observed ranking exhibits significant unfairness. 

\subsection{The re-ranking strategy}\label{sec:rerank}

If our ranking fails the fairness test for a given $\alpha$ (i.e., $\tilde{f}_k < \alpha$), it indicates that the ranking does not meet the fairness constraints. In such cases, we might decide that re-ranking becomes necessary to ensure compliance with fairness standards. The re-ranking strategy we adopt aligns with the FA*IR algorithm presented in \cite{zehlike2017fa}, with a modification to replace the binomial quantile function with the hypergeometric quantile function. This algorithm is proven to solve the fair top-$k$ ranking problem, satisfying four key properties:
\begin{enumerate}
    \item \textbf{In-group monotonicity:} The relative ordering of items within the same demographic group is preserved.
    \item \textbf{Ranked group fairness:} The ranking ensures fair representation of protected groups across the top-$k$ positions, satisfying the fairness constraints defined by $\alpha$.
    \item \textbf{Optimal selection utility:} Among all rankings that meet the fairness constraints, the algorithm selects the ranking that maximizes utility.
    \item \textbf{Maximised ordering utility:} Subject to the first three properties, the ranking is adjusted to retain the original order of candidates as much as possible.
\end{enumerate}

The re-ranking algorithm proceeds by sequentially checking fairness at each position in the ranking and making minimal adjustments whenever a position fails the fairness test. Although the process relies on quantile functions, the re-ranked output is unique for any fixed ranking and parameter $\alpha$, making the algorithm fully deterministic. Further implementation details are provided in Appendix~\ref{app:rerank}. Figure \ref{fig:algo} (c) illustrates how the re-ranking procedure modifies the top-$k$ ranking proportions when applied with $\alpha = 10\%$.

\subsection{University admission case study}\label{sec:university}

To illustrate our methodology and empirically highlight the differences between the hypergeometric and finite binomial models, we present a case study based on real-world data. The dataset contains records of students from an elite university that recently expanded low-income enrolment via a government scholarship program \cite{alvarez2022college}. Students in this dataset are ranked based on their GPA, and are characterized by additional attributes such as socioeconomic status (SES) and gender. 
We assess fairness by setting a significance level of $\alpha=10\%$ (a standard choice in statistical testing) across three distinct cohorts using different protected attributes, and show the results in Figure \ref{fig:uni}. In each case, we visualise the cumulative proportion of protected candidates across ranking positions, comparing it to the acceptance regions (i.e., confidence intervals at level $1-\alpha$ for $k$ tests) computed under both the hypergeometric and finite binomial models. Rankings that fall outside these intervals at any position $j \leq k$ are considered statistically unfair.

\paragraph{Cohort (a): Under-representation of low-SES students.}
The protected group in this cohort comprises students with low socioeconomic status. As shown in Figure \ref{fig:uni} (left), these students are under-represented---except at the very top---relative to their overall proportion ($58\%$). Therefore, to test whether this under-representation is statistically significant, we apply $k=n$ ($=56$) one-sided hypergeometric tests on the top-$j$ subsets of the ranking, for $j=1,\dots,n$, with a total error rate of $\alpha$ (with adjusted parameter $\alpha_c=1.88\%$).
The resulting fairness score is $4.11\%$, below $\alpha$, indicating unfairness. The finite binomial model also flags the ranking as unfair, but with a higher fairness score, suggesting a less extreme deviation. 

\paragraph{Cohort (b): Proportional representation of low-SES students.}
The protected group in this cohort also consists of low-SES students. In this case, they are represented in proportion to their share in the population ($32\%$), as shown in Figure \ref{fig:uni} (middle). Therefore, we test the significance of this representation using a two-sided fairness test over the full rank. Both the hypergeometric and finite binomial models yield fairness scores above $70\%$, exceeding $\alpha$, indicating no significant deviation and thus a fair outcome.

\paragraph{Cohort (c): Over-representation of female students.}
The protected group in this cohort consists of female students. 
As shown in Figure \ref{fig:uni} (right), they are over-represented across the ranking, except at the very top. When the ranking covers a large proportion of candidates ($\approx 70\%$ or more), their representation overlaps with the upper boundary. This means all female candidates have already appeared in higher positions, leaving the remainder of the ranking dominated by male candidates.
The hypergeometric model, which accounts for finite population effects, flags this as unfair, with a fairness score of $1.68\%$. In contrast, the finite binomial model does not detect unfairness, as it does not capture effects in the lower ranks.

\begin{figure*}[ht]
\begin{center}
\includegraphics[width=\textwidth]{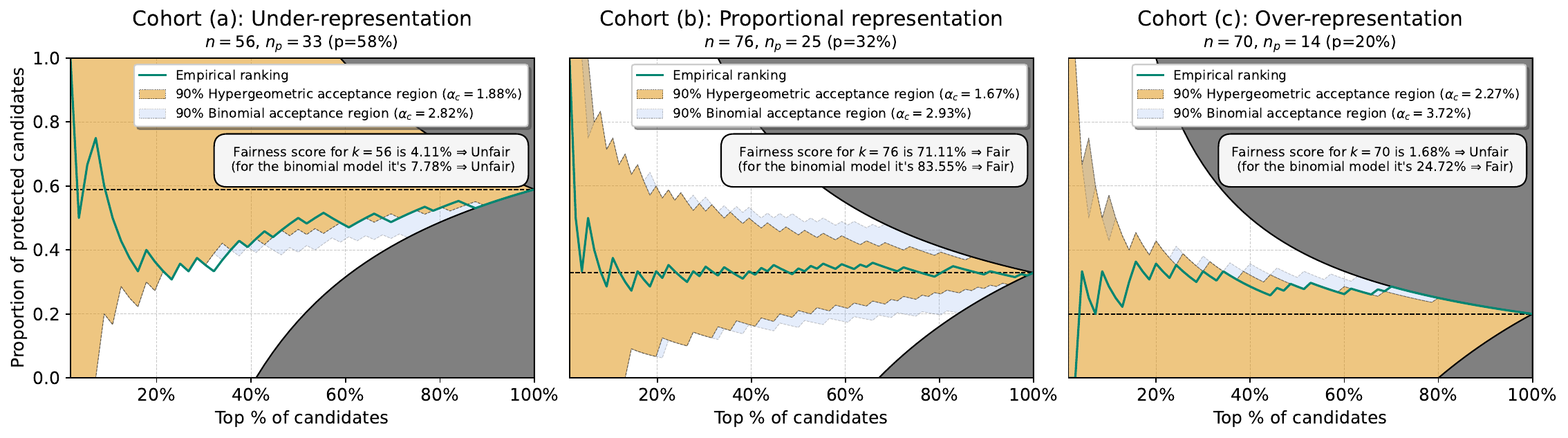}
\caption{\textbf{Fairness assessment of real-world university admission rankings using hypergeometric and finite binomial models.} The three panels show the cumulative proportion of protected candidates (green line) in the top positions of three different cohort rankings, compared to the acceptance regions (confidence intervals with level $1-\alpha = 0.9$) derived from the hypergeometric (orange) and finite binomial (light blue) models. In each case, the fairness score is reported for the full ranking ($k = n$). The left panel shows under-representation of low-SES students; both models flag it as unfair, though the finite binomial model yields a less extreme result. The middle panel shows an almost perfectly fair ranking. The right panel illustrates the over-representation of female candidates: when the ranking covers a large proportion of candidates ($\approx 70\%$) all female candidates have been selected, meaning that the bottom 30\% can only include male students. Here, only the hypergeometric model flags the ranking as unfair.}
\Description{Three plots comparing the proportion of protected candidates in university admission rankings to confidence intervals under hypergeometric and finite binomial models. Each panel corresponds to a different cohort. Panel (a) shows under-representation of low-SES students, flagged as unfair by both models. Panel (b) shows a balanced distribution for low-SES students, not flagged as unfair. In Panel (c) the hypergeometric model detects unfairness in the bottom ranks, while the finite binomial model isn't able to capture this.}
\label{fig:uni}
\end{center}
\end{figure*}

%% file: Input_files/affirmative_action.tex
\section{Beyond demographic parity}\label{sec:aa}

In our conference admission example in Section \ref{sec:fair_generative}, we assumed that the underlying goal was to achieve demographic parity with respect to the candidate pool. However, this is not always the objective in practice. In situations where the pool itself reflects underlying societal imbalances, interventions may need to extend beyond demographic representation. For example, in STEM disciplines, where the candidate pool is often predominantly male due to historical and structural inequalities, achieving gender representation closer to parity is often a common objective. Similarly, in international academic conferences, participants from low-income countries are frequently under-represented due to economic barriers or resource limitations. In such cases, we may prioritise the inclusion of these groups to promote equity and global representation, even when their share in the candidate pool is relatively small.

\subsection{Affirmative action policies} Selection goals often vary depending on the context, the characteristics of the group in question, and broader societal objectives, leading to diverse implementations of affirmative action (AA) policies. Affirmative action refers to `a set of ethically driven policies aimed at providing special opportunities to a historically disadvantaged group in order to enable the members of this group to compete with their more privileged peers' \cite{chowdhury2023heterogeneity}. 
The methods of implementation can differ in intensity, ranging from measures that indirectly influence outcomes to those that mandate specific results. \cite{harrison2006understanding} provide a comprehensive framework for understanding the range of affirmative action policies, with an application to the labour market. Such policies encompass `soft' measures that influence the distribution of outcomes indirectly---such as training, outreach programmes, or action plans that affect the chances of members of designated groups achieving certain outcomes---and also `hard' programmes in the form of institutionalised quotas or reservations that mandate the distribution of outcomes directly.

While affirmative action policies aim to reduce intergroup inequalities, they are not without criticism. Key concerns include the potential reinforcement of social divisions and stigmatisation by prioritising certain groups, the perception of reverse discrimination, and the risk of benefiting privileged individuals within target groups while neglecting intersecting inequalities. 
For a systematic review of both the benefits and challenges of AA policies, see \cite{schotte2023does}.


We now turn to the integration of affirmative action policies within fair generative processes. Such policies often aim to improve access for historically marginalised groups by setting admission quotas, awarding additional points in examinations, or lowering entrance thresholds \cite{schotte2023does, mathioudakis2020affirmative}. In this work, we focus specifically on proportional quotas, which aim to keep the proportion of protected individuals in the selected group close to a predetermined target. In the context of university admissions discussed in Section~\ref{sec:university}, these quotas often take the form of explicit seat reservations based on caste, race, or socioeconomic status \cite{gururaj2021affirmative}.


\subsection{A weighted hypergeometric model}

 In the finite binomial model defined in Section \ref{sec:fair_generative}, affirmative action policies, and in particular quotas, can be implemented by fixing the sampling probability of the protected group to a certain target proportion $\rho$. This ensures that the model behaves as if each top-$k$ ranking prefix has an expected proportion of $\rho$ protected individuals, as long as the candidate pool allows for it. By contrast, the hypergeometric model introduced earlier naturally converges to demographic parity, limiting its flexibility in enforcing quotas that differ from the demographic proportions.

To address this limitation, we can extend the hypergeometric model by introducing weights to the sampling process. Specifically, we define a relative weighting factor, or \emph{odds ratio}, $\omega$, which determines the relative likelihood of selecting a protected individual over a non-protected one when both are present in the candidate pool. This adjustment allows us to control the representation of protected individuals and align it with predefined quotas. 

\begin{definition}[Non-central hypergeometric model]
A ranking $(x_1, x_2, \dots, x_n)$ with $n_p$ protected elements is said to be generated under the \emph{non-central hypergeometric model} if we draw each element in the ranking based on weighted probabilities from a population of size $n$ that contains exactly $n_p$ protected individuals. Specifically, the probability of selecting a protected individual is determined by a weighting factor, or \emph{odds ratio}, $\omega$, which represents the relative likelihood of selecting a protected individual compared to a non-protected one. The sampling is performed without replacement, and the probability of selecting a protected individual adjusts dynamically based on the remaining population. 

When $\omega > 1$, protected individuals are favoured, increasing their likelihood of being sampled into higher-ranked positions. Conversely, when $\omega < 1$, non-protected individuals are favoured. When $\omega = 1$, the sampling reduces to the classical hypergeometric model. Let $Y_k=\sum_{i=1}^k X_i$, representing the cumulative number of protected individuals in the top $k$ positions. Then $Y_k$ follows a Wallenius' non-central hypergeometric distribution \cite{wallenius1964biased} of the form $\emph{WncHyp}(n,n_p,k,\omega)$.
\end{definition}

\paragraph{How to choose the odds ratio.} To ensure that the representation of the protected group matches the desired proportion $\rho$, we need to choose the correct value of $\omega$. Specifically, we choose the value for which the odds of picking a protected candidate in the first draw is equal to $\rho$. This implies setting:
\[\mathbb P(X_1=1)=\frac{\omega n_p}{\omega n_p + (n-n_p)} = \rho\]
Rearranging this equation gives the formula for $\omega$:
\[ \omega = \frac{\rho}{1-\rho}\cdot \frac{n-n_p}{n_p} = \frac{\rho}{1-\rho}\cdot \frac{1-p}{p}\]
where $p=\frac{n_p}{n}$ is the proportion of protected individuals in the candidate pool. It is evident that when $\rho=p$---corresponding to demographic parity---the odds ratio simplifies to $\omega=1$, restoring the classical hypergeometric model. Additionally, the binomial model with sampling probability $f=\rho$ can be viewed as a weighted sampling model with the same value of $\omega$, but with replacement instead of without replacement.

\begin{figure}[ht]
\begin{center}
\includegraphics[width=0.45\textwidth]{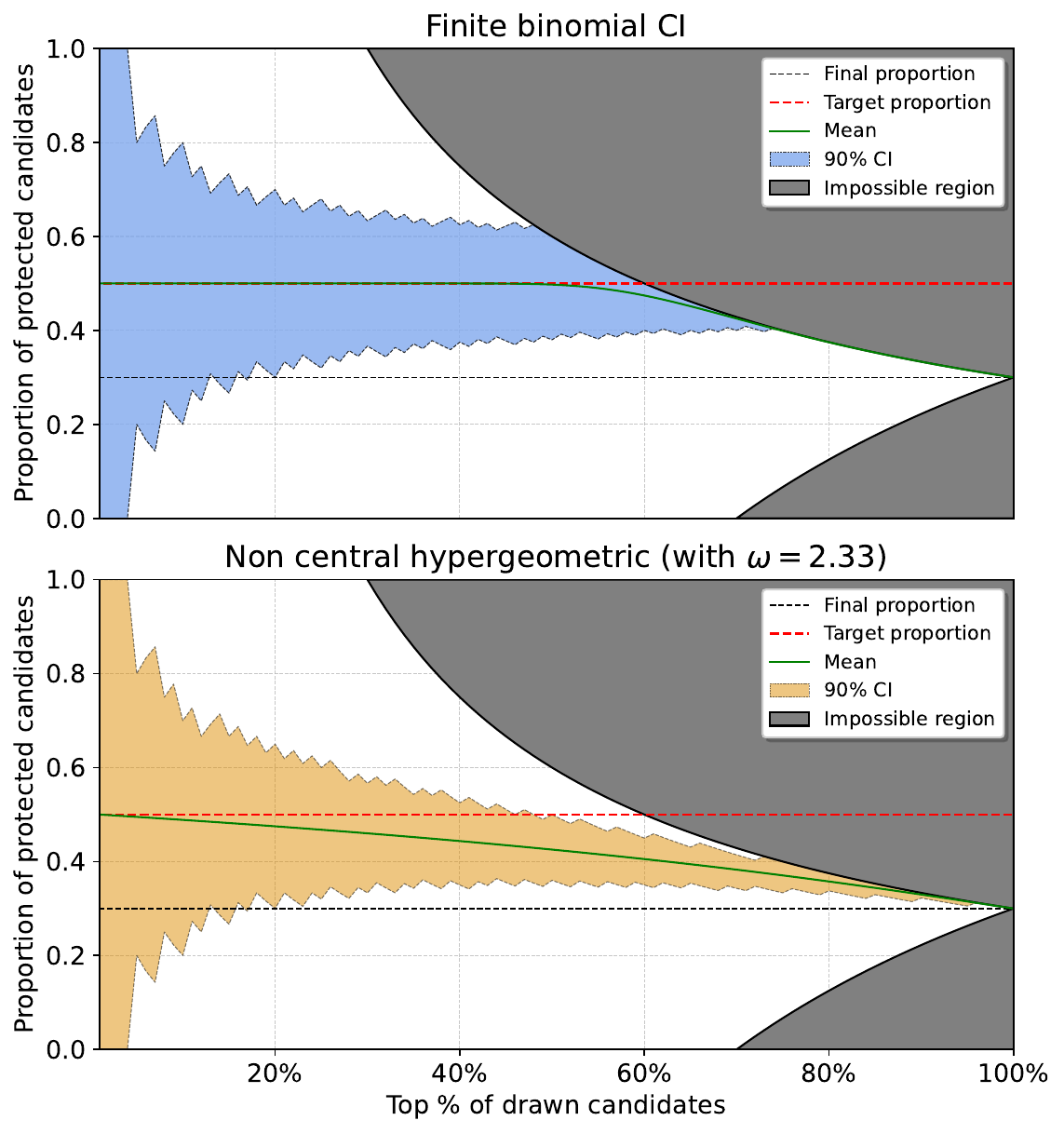}
\caption{\textbf{Binomial vs. Non-central hypergeometric confidence intervals.} Two-sided $90\%$ confidence intervals for the proportion of protected individuals drawn from the finite binomial and the non-central hypergeometric distributions, with $n=100$, $n_p=30$, $\omega=2.33$, target proportion $\rho=0.5$ and fairness probability $f$ equal to $\rho$. The plot visualizes the distribution of possible proportions of protected individuals as the percentage of drawn individuals increases.}
\Description{ Figure comparing two-sided 90\% confidence intervals for the proportion of protected individuals selected under the finite binomial and non-central hypergeometric models. The x-axis shows the cumulative percentage of individuals drawn from a population of size $n=100$, with 30 protected individuals. The target fairness proportion is $\rho=0.5$, and the odds ratio in the hypergeometric model is $\omega=2.33$, with a fairness probability $f$ equal to $\rho$. The finite binomial model enforces the target proportion strictly across all sample sizes, resulting in symmetric and narrower confidence intervals that remain centred at $\rho$. In contrast, the non-central hypergeometric model achieves the target proportion at the top of the ranking but gradually converges to demographic parity as more candidates are drawn. Its confidence intervals are wider and asymmetric.}
\label{fig:diagram_aa}
\end{center}
\end{figure}

\paragraph{The differences between the models.} The finite binomial model strictly enforces the target proportion, ensuring this proportion is maintained throughout the sampling process, as long as there are enough candidates from the required group. In contrast, the non-central hypergeometric model ensures the target proportion $\rho$ is achieved at the first draw but gradually converges to demographic parity as additional individuals are sampled. 

This difference is evident in Figure \ref{fig:diagram_aa}, where the finite binomial model (upper plot) maintains strict alignment with the target proportion, while the hypergeometric model (lower plot) exhibits a smoother convergence toward parity. Although the non-central hypergeometric model does not strictly adhere to the target proportion $\rho$, it reduces the risk of boundary effects where the remaining candidates in the ranking belong solely to one group. In the finite binomial model, such boundary scenarios are very likely when sampling a large proportion of individuals.

When $\rho>p$ (as in the plot above), the hypergeometric approach reduces the risk of reverse discrimination against the non-protected group, thus avoiding the marginalisation of other disadvantaged groups that do not belong to the protected category. Conversely, when $\rho<p$, it mitigates the risk of under-representation of the protected group, which may arise if quotas are deliberately set below the overall proportion of the protected group. For instance, a policy imposing a fixed quota for women in leadership positions may result in under-representation if women comprise a larger proportion of the candidate pool \cite{LSE2014}. By addressing both reverse discrimination and under-representation, the hypergeometric model may increase the acceptability and likelihood of enacting affirmative action policies that aim for equitable representation.


Finally, our framework enables a quantitative assessment of how different quota policies impact both protected and non-protected groups. First, it allows us to assess whether a target proportion is statistically plausible under fair random selection, helping justify quota decisions. Second, we can compare the effects of strict quotas (enforcing a minimum $\rho$) versus softer quotas (achieving $\rho$ in expectation at the first draw). Third, it quantifies group advantage by estimating the most likely weight $\omega$ under which the observed ranking would be considered fair.

%% file: Input_files/conclusions.tex
\section{Conclusions}\label{sec:disc}

In this study, we introduced \hyperfair, a framework for assessing and enforcing fairness in ranking systems. By leveraging the hypergeometric and non-central hypergeometric distributions, our approach addresses critical challenges in achieving fair representation within finite candidate pools. Unlike traditional binomial models \cite{yang2017measuring}, \hyperfair~better reflects real-world constraints, particularly in scenarios where the number of protected candidates is small or the proportion of selections is large relative to the population. 
We also demonstrated that the hypergeometric model inherently accounts for fairness across the bottom positions, offering a more comprehensive perspective on representation throughout the entire rank. These benefits were illustrated analytically and through a real-world case study on university admissions.

We additionally introduced a weighted generative model that supports the implementation and evaluation of affirmative action policies, allowing target quotas to be enforced at the top of the rank while gradually converging toward demographic parity as more candidates are selected. This ensures equitable representation throughout the entire rank and enhances the public acceptability of fairness interventions.
To support efficient fairness evaluations, we also developed a Monte Carlo-based algorithm capable of handling multiple statistical tests across different positions in the ranking. 
Finally, throughout the paper we identified the fundamental differences between the hypergeometric and binomial models. In Appendix \ref{app:choice}, we provide guidelines on the conditions under which each model is more suitable for fairness evaluations, in contexts where demographic parity is the objective.  


While \hyperfair~provides a powerful framework for evaluating fairness in rankings, it is important to emphasise its limitations. 
First, fairness interventions should not be applied in isolation. Methods such as re-ranking or affirmative action require human oversight to ensure that algorithmic decisions align with ethical and contextual considerations.
Second, our approach is only applied to pre-computed rankings and does not mitigate biases that may arise during earlier phases of candidate evaluation.

Future work should extend this framework to multiclass settings and incorporate multiple attributes to address intersectional discrimination. Another important direction is to integrate the underlying scores (e.g., GPAs) that generated the ranking into the sampling model. Moving beyond relative ordering enables more accurate fairness assessments---for instance, preventing large exposure differences between similarly scored candidates---and supports alternative fairness definitions, such as equal opportunity.
Finally, our model offers a promising basis for evaluating the validity of other existing fairness metrics, such as those presented in \cite{schumacher2024properties}.

To conclude, this work advances the methodological foundations of fairness in ranking by introducing models that remain robust even when the share of selected candidates is large or the overall pool is imbalanced---a common, yet under-addressed challenge. By fostering rigorous fairness assessments, we lay the foundation for developing fairer ranking algorithms, and in turn, fairer search engines and recommendation systems.

\paragraph{Reproducibility.} All code used in this paper is available online at \url{https://github.com/CSHVienna/hyper_fair}.

\begin{acks}
    This research work was funded by the European Union under the Horizon Europe MAMMOth project, Grant Agreement ID: 101070285. L.E.N. received support from the Vienna Science and Technology Fund WWTF under project No. ICT20-079. We are also grateful to the anonymous reviewers for their suggestions on this paper.



    \paragraph{Empirical data.} We thank Professor María José Álvarez-Rivabulla and her collaborators for providing the anonymized student data. To safeguard students' privacy, the data is not publicly available.  Inquiries about data collection, access, or maintenance should be directed to \href{mailto:mj.alvarez@uniandes.edu.co}{mj.alvarez@uniandes.edu.co}.
\end{acks}

%% file: Input_files/appendix.tex
\appendix

\section{Boundaries of ranking proportions}\label{app:boundaries}

In this section, we explain the derivation of the boundaries of the impossible regions shown in all the figures. As discussed in Section \ref{sec:fair_generative}, these boundaries correspond to two extreme scenarios. Let $x = \frac{k}{n}$ represent the proportion of individuals considered up to rank $k$. The first scenario assumes all protected candidates occupy the top positions. We will demonstrate why, in this case, $p_k = 1$ until $x = p$, after which $p_k$ follows the curve $\frac{p}{x}$.

To begin, recall that the population consists of $n_p$ protected candidates and $n - n_p$ non-protected candidates, with $p = \frac{n_p}{n}$ denoting the proportion of protected candidates in the population. The proportion of protected individuals observed up to position $k$ is given by $p_k = \frac{y_k}{k}$, where $y_k$ is the cumulative number of protected individuals selected up to rank $k$.

In the extreme scenario where all protected candidates occupy the top ranks, the first $n_p$ positions are filled exclusively with protected candidates. Therefore, until $k = n_p$ (i.e., when $x = \frac{n_p}{n} = p$), $p_k = 1$, as all selected individuals are protected.

Beyond this point ($k > n_p$), no more protected candidates can be selected. The proportion $p_k$ will then decrease as more non-protected individuals are added to the ranking. To determine how $p_k$ decreases, suppose $t$ additional individuals are selected. In this case, $x = \frac{k + t}{n}$, and the proportion of protected candidates becomes:
\[
p_k = \frac{n_p}{k + t} = \frac{n_p}{n} \cdot \frac{n}{k + t} = \frac{p}{x}.
\]

Using a similar reasoning, we can analyse the second extreme scenario, where all protected candidates are relegated to the lowest positions. In this case, $p_k = 0$ until $x = 1 - p$, as no protected individuals are selected before this threshold. After this point, $p_k$ increases according to the curve $1 - \frac{1 - p}{x}$.

\begin{figure}[ht]
\begin{center}
\includegraphics[width=0.48\textwidth]{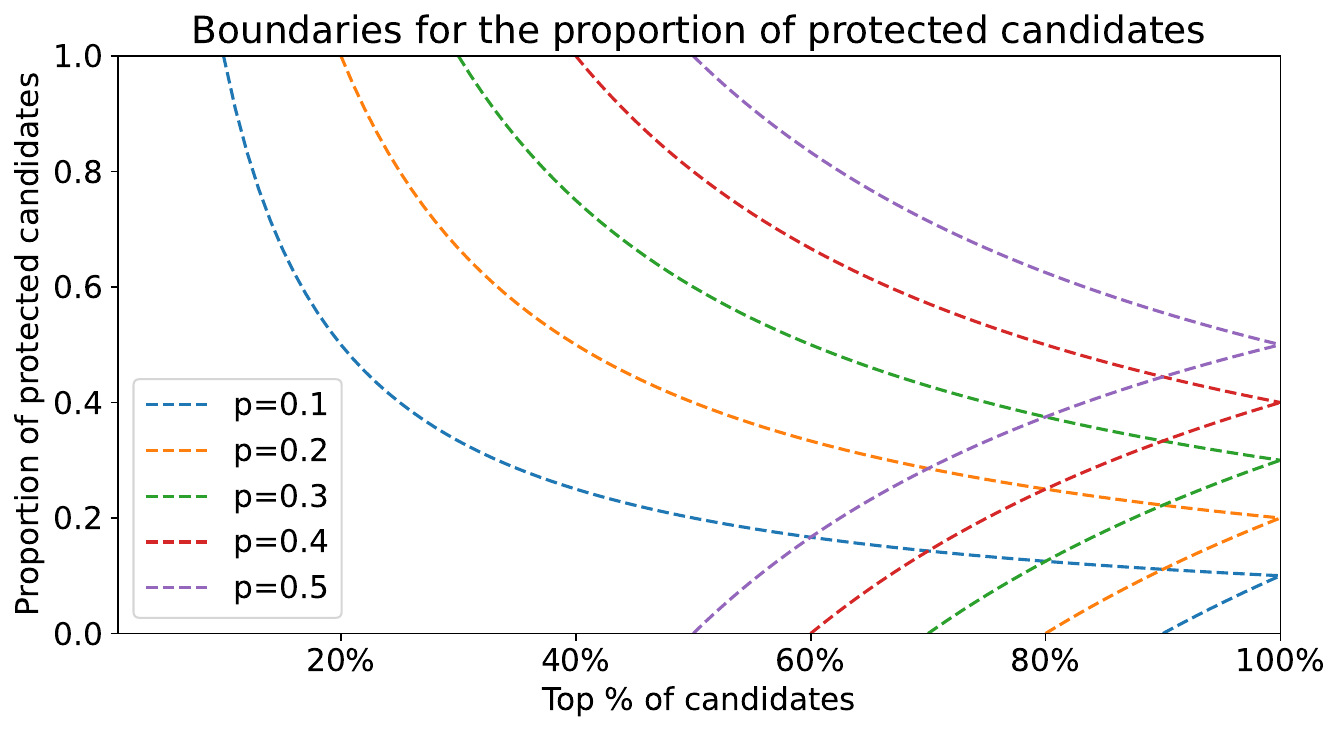}
\caption{\textbf{Mathematical boundaries for different values of the total proportion of protected candidates $p$.} }
\Description{ Figure showing boundary regions for fairness evaluations at varying values of the protected group proportion $p$, ranging from 0.1 to 0.5.}
\label{fig:boundaries}
\end{center}
\end{figure}

Figure \ref{fig:boundaries} illustrates the boundaries for different values of $p$. For small values of $p$, these boundaries play a critical role, as they restrict the range of achievable proportions in the plot, highlighting the importance of considering such constraints when analysing ranking fairness.

\section{Proof of fairness in the bottom positions}\label{app:proof}

\begin{proposition}
\label{prop:bottom_fairness}
Let $ H_k \sim \text{Hyp}(n, n_p, k) $ and $ W_{n-k} = n_p - H_k \sim \text{Hyp}(n, n_p, n - k) $ be the number of protected individuals among the top-$ k $ and bottom-$ n - k $ positions of a ranking, respectively. Then, under a two-sided fairness test at significance level $ \alpha $, the ranking passes the fairness test for the top-$ k $ if and only if it passes the fairness test for the bottom-$ n - k $.
\end{proposition}

\begin{proof}
Let $ H_k \sim \text{Hyp}(n, n_p, k) $ represent the number of protected individuals in the top-$ k $ positions, and let $ W_{n-k} = n_p - H_k $ represent the number of protected individuals in the remaining $ n - k $ positions. By the properties of the hypergeometric distribution, we have:
\[
W_{n-k} \sim \text{Hyp}(n, n_p, n - k).
\]

Let $ Q_{\gamma}^{(k)} $ denote the $ \gamma $-quantile of the distribution of $ H_k $, and $ Q_{\gamma}^{(n-k)} $ denote the $ \gamma $-quantile of the distribution of $ W_{n-k} $. The two-sided fairness test accepts the ranking at level $ \alpha $ if:
\[
Q_{\alpha/2}^{(k)} \leq H_k \leq Q_{1 - \alpha/2}^{(k)}.
\]

Rewriting in terms of $ W_{n-k} = n_p - H_k $, we get:
\[
n_p - Q_{1 - \alpha/2}^{(k)} \leq W_{n-k} \leq n_p - Q_{\alpha/2}^{(k)}.
\]

Due to the complementarity of the hypergeometric distribution, the quantiles satisfy:
\[
Q_{\alpha/2}^{(n-k)} = n_p - Q_{1 - \alpha/2}^{(k)} \quad \text{and} \quad Q_{1 - \alpha/2}^{(n-k)} = n_p - Q_{\alpha/2}^{(k)}.
\]

Substituting these back yields:
\[
Q_{\alpha/2}^{(n-k)} \leq W_{n-k} \leq Q_{1 - \alpha/2}^{(n-k)},
\]
which is precisely the condition that the bottom-$ n - k $ positions pass the two-sided fairness test. Therefore, the fairness of the top-$ k $ positions is equivalent to the fairness of the bottom-$ n - k $ positions under this test.
\end{proof}

\section{Complexity and convergence of Monte Carlo algorithm}\label{app:convergence}

In this section, we analyse the computational complexity and convergence properties of the Monte Carlo algorithm used to compute the adjusted $\alpha_c$ in the context of multiple statistical tests.

The Monte Carlo algorithm can operate in two modes depending on how the cumulative distribution function (CDF) of the hypergeometric distributions are evaluated:

\begin{figure*}[ht]
\begin{center}
\includegraphics[width=0.9\textwidth]{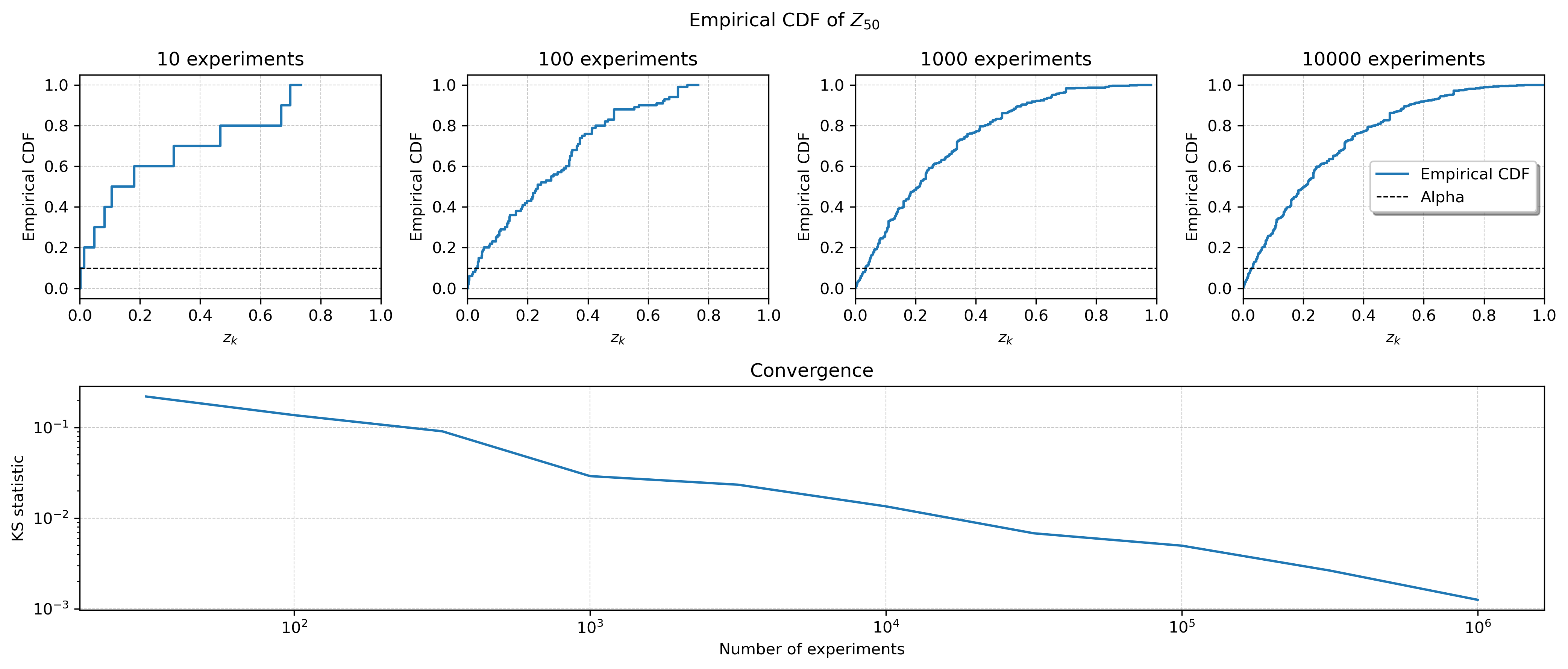}
\caption{\textbf{Convergence of the ECDF of $Z_k$ used in the Monte Carlo algorithm.} }
\Description{Plot showing the convergence of the empirical cumulative distribution function of $Z_k$ used in the Monte Carlo algorithm based on the number of experiments $N_e$. First, four plots showcase the behaviour of the ECDF with different number of experiments, respectively 10, 100, 1000 and 10000, showcasing a slow convergence to a smooth line. Below, another panel showcases the convergence in terms of the Kolmogorov-Smirnov distance between more complex models, indicating a smooth and linear convergence in the log-log plot.}
\label{fig:convergence}
\end{center}
\end{figure*}

\begin{itemize}
    \item \textbf{Using the Analytical CDF:} When the hypergeometric CDF can be computed directly or is precomputed, the computational cost of the Monte Carlo algorithm is $O(N_e \cdot k)$, where $N_e$ is the number of random rankings generated and $k$ is the number of positions in the ranking. The algorithm is detailed in Algorithm 1.
    \item \textbf{Using the Empirical CDF (ECDF):} If the computational cost of the hypergeometric CDF is prohibitively high, the algorithm estimates the CDF empirically using sampled rankings. This involves generating rankings and calculating quantiles based on the empirical distribution. The computational cost in this case is $O(N_e \cdot k + N_e \log(N_e) \cdot k)$, where the additional term arises from sorting the samples for quantile estimation.
\end{itemize}

\begin{algorithm}\label{alg:monte_carlo_cdf}
\caption{Monte Carlo fairness adjustment (with analytical CDF)}
\KwData{$n, n_p, k, \alpha, N_e$} \LinesNumbered
\KwResult{$\alpha_c$} 

$z \gets [0]$ \Comment*[r]{Initialize a vector of length $N_e$}
\For{$i = 1$ to $N_e$}{
    $x \gets \text{permutation}(n, n_p, k)$ \Comment*[r]{Generate a random ranking permutation}
    $ecdf \gets [0]$ \Comment*[r]{Initialize ECDF values for $k$ positions}
    $y \gets 0$\;
    \For{$j = 1$ to $k$}{
        $y \gets y + x[j]$ \Comment*[r]{Compute cumulative sum}
        $ecdf[j] \gets \text{hypergeometric\_cdf}(y, n, n_p, j)$\;
    }
    $z[i] \gets \min(ecdf)$\;
}
$z \gets \text{sort}(z)$\;
$\alpha_c \gets z[\lfloor \alpha \cdot N_e \rfloor]$\;
\end{algorithm}

The computational cost of the Monte Carlo algorithm critically depends on the number of sampled rankings $N_e$. $N_e$ must be sufficiently large to ensure that the empirical cumulative distribution function (ECDF), denoted as $F_{N_e}(x)$, converges to the true CDF $F(x)$.

The Glivenko-Cantelli theorem \cite{glivenko1933sulla, van2000asymptotic} guarantees that the ECDF $F_{N_e}(x)$ converges uniformly to the true CDF $F(x)$ as $N_e \to \infty$, with probability 1. However, the rate of this convergence is essential to quantify for practical purposes, and it is provided by the Dvoretzky–Kiefer–Wolfowitz (DKW) inequality \cite{dvoretzky1956asymptotic, van2000asymptotic}. The DKW inequality states that, with probability $1 - \beta$, the true CDF $F(x)$ lies within a margin $\varepsilon$ of the ECDF $F_{N_e}(x)$:
\[
F_{N_e}(x) - \varepsilon \leq F(x) \leq F_{N_e}(x) + \varepsilon,
\]
where $\varepsilon = \sqrt{\frac{\ln(2/\beta)}{2N_e}}$.

To relate this to practical requirements, let us define $N_e = 10^\gamma$ as the number of random rankings and $\varepsilon = 10^{-\delta}$ as the desired error margin for approximating the CDF of $Z_k$. The relationship becomes:
\[
\gamma = \log_{10}\left(\frac{\ln(2/\beta)}{2 \cdot 10^{-2\delta}}\right) = \log_{10}\left(\frac{\ln(2/\beta)}{2}\right) + 2\delta.
\]
Assuming a confidence level of $1 - \beta = 0.9$ (i.e., $\beta = 0.1$), and substituting $\ln(2/\beta) \approx 2.3$, we find $\gamma \approx 0.175 + 2\delta$. For example, if we require an error margin of $\varepsilon = 10^{-3}$ ($\delta = 3$), the number of samples needed is:
\[
\gamma = 0.175 + 2(3) = 6.175, \quad \text{or equivalently,} \quad N_e \approx 10^{6.175} \approx 1.5 \times 10^6.
\]
This convergence also appears empirically, as illustrated by Figure \ref{fig:convergence} for the ECDF of $Z_{50}$ with $n=100,n_p=30$. The panels above illustrate the convergence of the ECDF for an increased number of experiments visually, while the bottom plot shows the convergence of the Kolmogorov-Smirnov (KS) statistic \cite{an1933sulla} for an increasing number of experiments. We define the KS statistic between consecutive settings that contain more experiments, specifically $D_{10^\gamma} = \sup_x|F_{10^\gamma}(x)-F_{10^{\gamma-1}}(x)|$.

The accuracy of the quantile estimates for $Z_k$ depends directly on the convergence of the ECDF $F_{N_e}(x)$ to the true CDF $F(x)$. Since quantiles are defined as the inverse of the CDF, any deviation in the ECDF affects the estimated quantiles. As $N_e$ increases, the error margin $\varepsilon$ decreases, ensuring more precise quantile estimates. Thus, the rate of convergence of the ECDF governs the reliability of quantile-based calculations.

Since we can assume that $N_e$ is a constant (around one million) we can notice that in comparison, the computational cost of the algorithm to find the adjusted parameter $\alpha_c$ in \cite{zehlike2017fa, zehlike2020note} is $O(k^2\log(k))$, and is therefore more computationally expensive when $k$ is large.

\input{Input_files/rerank}

\section{How to choose the correct generative model}\label{app:choice}

\paragraph{When to use the hypergeometric model:}
\begin{enumerate}
    \item \textbf{When $ k $ is large relative to $ n $:}  
    When the number of sampled candidates $ k $ is significant compared to the total population size $ n $, the finite nature of the sampling process becomes crucial. The hypergeometric distribution provides an exact model for fairness analysis in such settings, whereas the binomial model would overestimate the variance by ignoring the limited size of the population.  
    \begin{itemize}
        \item \textit{Example:} Suppose there are 50 applicants for a job ($ n = 50 $), and a company shortlists the top 30 candidates ($ k = 30 $). In this case, the sampling covers a large fraction of the population, and fairness evaluations based on the hypergeometric model are essential to account for the finite population size.
    \end{itemize}

    \item \textbf{When $ n_p $ is small:}  
    If the number of protected individuals $ n_p $ is small compared to the total population $ n $, the hypergeometric model accurately reflects the finite group dynamics. Ignoring these effects can lead to significant fairness misjudgments for under-represented groups.  
    \begin{itemize}
        \item \textit{Example:} In a scholarship program with 100 applicants, where only 5 belong to a protected group ($ n_p = 5 $), sampling 15 candidates ($ k = 15 $) may have a strong impact on the representation of the protected group. The hypergeometric model correctly incorporates the finite number of protected individuals.
    \end{itemize}

    \item \textbf{When fairness at the bottom-$ k $ positions is critical:}  
    The hypergeometric model naturally accommodates fairness evaluations across all levels of the ranking, including the bottom positions. This is especially relevant in applications where individuals at lower-ranking positions are still considered for fallback opportunities or face substantial consequences of bias.  
    \begin{itemize}
        \item \textit{Example:} In housing allocation programs, fairness at the bottom ranks ensures that under-represented groups are not excluded from fallback options, such as waiting lists. Similarly, in public service delivery programs, individuals at the bottom of a ranked list may face the most significant consequences of under-representation, making fairness at these positions crucial.
    \end{itemize}
\end{enumerate}

\paragraph{When to use the binomial model:}
\begin{enumerate}
    \item \textbf{When $ n $ is unknown or very large:}  
    In situations where the total population size $ n $ is either unknown or extremely large, the binomial distribution provides a practical approximation. This is because the binomial model assumes sampling with replacement, removing the need for knowledge of $ n $.  
    \begin{itemize}
        \item \textit{Example:} Consider evaluating fairness in the top-10 results of a search engine query. Here, the total number of webpages indexed by the search engine ($ n $) is enormous and irrelevant to the analysis. The binomial model simplifies the problem by focusing only on the top results.
    \end{itemize}
\end{enumerate}

%% file: Input_files/rerank.tex
\section{The re-ranking strategy}\label{app:rerank}

In this section, we describe in more detail the re-ranking algorithm that we use in this paper when one given ranking doesn't satisfy our fairness criteria.

\begin{algorithm}
\label{alg:rerank}
\caption{Re-ranking Strategy}
\KwData{$\tilde{x}, \alpha_c, id$} \LinesNumbered
\KwResult{Re-ranked $\tilde{x}$ and $id$} 

$P\_prot \gets [0]$ \Comment*[r]{Priority queue for protected group members}
$m \gets [0]$ \Comment*[r]{Vector to store fairness thresholds for each position}
\For{$i = 1$ to $n$}{
    \If{$\tilde{x}[i]=1$}{
        add($P\_prot$, $i$) \Comment*[r]{Add index of protected group member to queue}
    }
    $m[i] \gets F^{-1}(\alpha_c, n, n_p, i)$ \Comment*[r]{Compute fairness threshold}
}
$y \gets 0$ \Comment*[r]{Initialize cumulative count of protected group members}
\For{$i = 1$ to $n$}{
    $y \gets y + \tilde{x}[i]$ \Comment*[r]{Update count of protected members}
    \uIf{$y < m[i]$ }{ 
        $ind \gets $ dequeue($P\_prot$) \Comment*[r]{Get next protected group member}
        $\tilde{x}[i] \gets 1$ \Comment*[r]{Insert protected member at position $i$}
        $\tilde{x}[ind] \gets 0$ \Comment*[r]{Remove the replaced element}
        $id\_temp \gets id[i:ind-1]$ \Comment*[r]{Store intermediate IDs}
        $id[i] \gets id[ind]$ \Comment*[r]{Update ID at position $i$}
        $id[i+1:ind] \gets id\_temp$ \Comment*[r]{Shift other IDs}
        $y \gets y + 1$ \Comment*[r]{Increment count of protected members}
    }
    \ElseIf{$\tilde{x}[i]=1$}{
        dequeue($P\_prot$) \Comment*[r]{Remove protected group member from queue}
    }
}
\end{algorithm}

The re-ranking process operates by iteratively evaluating the fairness of the ranking at each position and making minimal adjustments to restore fairness whenever the test is failed. Algorithm 2 describes this approach in detail, where elements from the protected group are dynamically substituted into the ranking to satisfy the fairness constraints.

The algorithm takes as input two vectors: the binary vector indicating the demographic group of each individual in the ranking, $\tilde{x}$ and a vector containing the unique identifiers of individuals in the ranking, $id$. Additionally, the adjusted threshold derived from the Monte Carlo algorithm, $\alpha_c$, is given as input. 

The algorithm ensures that the protected group is appropriately represented in the ranking while preserving as much of the original order as possible. It achieves this by:
\begin{enumerate}
    \item \textbf{Checking fairness constraints:} At each position in the ranking, the algorithm verifies whether the cumulative representation of the protected group meets the fairness threshold.
    \item \textbf{Adjusting the ranking:} If the fairness test fails at any position, the algorithm substitutes the current element with the next available member of the protected group from a priority queue. The substitution involves reordering the subsequent elements to maintain consistency.
    \item \textbf{Preserving utility information:} The algorithm does not explicitly use utility values $y$, as the initial order of the ranking is assumed to reflect the utility. The re-ranking process relies solely on this relative order to make adjustments.
\end{enumerate}

In our discussion and algorithm, we do not explicitly mention the utility values $y$. This is because we assume that the order of the ranking is already based on utility, and this relative order is the only information required for re-ranking. Thus, the algorithm focuses solely on adjusting the ranking to meet fairness constraints without recalculating or explicitly using utility values.

Moreover, the focus of this algorithm is on addressing the under-representation of the protected group. However, the framework can be easily generalised to address other fairness scenarios, such as correcting over-representation or ensuring equal representation for both groups. By adapting the fairness thresholds and the re-ranking logic accordingly, the algorithm can be tailored to meet diverse fairness objectives in ranking systems.